\documentclass[11pt]{article}
\usepackage[top=3cm, bottom=2cm, left=3.5cm, right=3.5cm, includefoot]{geometry}

\usepackage{authblk}

\newif\ifdraft\draftfalse

\usepackage{color}                  
\usepackage{graphicx}               
\usepackage{amssymb}                           
\usepackage{amsmath}        
\usepackage{amsthm}                   
\usepackage{url}                               

\newtheorem{theorem}{Theorem}
\newtheorem{lemma}{Lemma}
\newtheorem{definition}{Definition}
\newtheorem{example}{Example}

\usepackage[linesnumbered,vlined,ruled,algo2e]{algorithm2e}

\usepackage[subrefformat=parens]{subcaption}
\newcommand{\rmq}[2]{\mathtt{rmq}(#1,#2)}
\newcommand{\lca}[2]{\mathtt{lca}(#1,#2)}
\newcommand{\rmqone}[2]{\mathtt{rmq1}(#1,#2)}

\newcommand{\substr}[3]{#1[#2:#3]}
\newcommand{\sublen}[3]{#1\langle #2,+#3\rangle}
\newcommand{\sublastlen}[3]{#1\langle #2,-#3\rangle}

\newcommand{\prefix}[2]{#1[:#2]}
\newcommand{\suffix}[2]{#1[#2:]}
\newcommand{\len}[1]{|#1|}
\newcommand{\reverse}[1]{{#1}^\mathrm{R}}

\newcommand{\match}[3]{\mathit{Match}({#1}, {#2}, {#3})} 

\newcommand{\co}{C}
\newcommand{\cop}[2]{\co[#1,#2]}

\newcommand{\cmaxijk}[3]{\mathit{M}[#1,#2]}
\newcommand{\opLCE}[2]{\mathit{opLCE}[#1, #2]}
\newcommand{\opLCEps}[4]{\mathit{opLCE}_{#3,#4}[#1, #2]}
\newcommand{\opL}{\mathit{opLCE}}
\newcommand{\opLq}[2]{\mathtt{opLCE}(#1, #2)}

\newcommand{\A}{A_{i,j}}

\newcommand{\orderisomorphic}{\approx}

\newcommand{\stTop}[1]{{#1}.\mathtt{top}()}
\newcommand{\stPop}[1]{{#1}.\mathtt{pop}()}
\newcommand{\stPush}[2]{{#1}.\mathtt{push}(#2)}

\newcommand{\problemabbr}{{LCS$_{k^{+}}$}}
\newcommand{\problemabbrk}[1]{{LCS$_{#1^{+}}$}}
\newcommand{\opproblem}{longest common subsequence in at least $k$ length order-isomorphic substrings}
\newcommand{\opproblemabbr}{{op-LCS$_{k^{+}}$}}
\newcommand{\opproblemabbrk}[1]{{op-LCS$_{#1^{+}}$}}

\newcommand{\Prev}[1]{\mathit{Prev}_{#1}}
\newcommand{\Next}[1]{\mathit{Next}_{#1}}

\newcommand{\Z}[1]{\mathit{Z}_{#1}}
\newcommand{\Zi}[2]{\Z{#1}[#2]}

\DeclareMathOperator*{\argmin}{argmin}
\DeclareMathOperator*{\argmax}{argmax}

\newcommand{\cpp}{C\nolinebreak\hspace{-.05em}\raisebox{.4ex}{\tiny\bf +}\nolinebreak\hspace{-.10em}\raisebox{.4ex}{\tiny\bf +}}

\newcommand{\PZS}{P{\v{Z}}{\v{S}}}

\usepackage{microtype}

\begin{document}

\title{Longest Common Subsequence in \\ at Least $k$ Length Order-Isomorphic Substrings\thanks{
	The final publication is available at Springer via \texttt{http://dx.doi.org/10.1007/978-3-319-51963-0\_28}.
	However, it contains some crucial typos (see Appendix).
}}

\author[1]{Yohei Ueki}
\author[1]{Diptarama}
\author[1]{Masatoshi Kurihara}
\author[2]{Yoshiaki Matsuoka}
\author[1]{Kazuyuki~Narisawa}
\author[1]{Ryo Yoshinaka}
\author[2]{Hideo Bannai}
\author[2]{Shunsuke Inenaga}
\author[1]{Ayumi~Shinohara}

\affil[1]{Graduate School of Information Sciences, Tohoku University, Sendai, Japan \newline
	\texttt{\{yohei\_ueki,diptarama,masatoshi\_kurihara\}@shino.ecei.tohoku.ac.jp \newline
		\{narisawa,ry,ayumi\}@ecei.tohoku.ac.jp}}
\affil[2]{Department of Informatics, Kyushu University, Fukuoka, Japan \newline
	\texttt{\{yoshiaki.matsuoka, bannai, inenaga\}@inf.kyushu-u.ac.jp}}

\pagestyle{plain} 

\date{}

\maketitle

\begin{abstract}
We consider the longest common subsequence (LCS) problem with the restriction that 
the common subsequence is required to consist of at least $k$ length substrings.
First, we show an $O(mn)$ time algorithm for the problem
which gives a better worst-case running time than existing algorithms,
where $m$ and $n$ are lengths of the input strings.
Furthermore, we mainly consider the LCS in at least $k$ length order-isomorphic substrings problem.
We show that the problem can also be solved in $O(mn)$ worst-case time
by an easy-to-implement algorithm.

\end{abstract}

\section{Introduction}
The \emph{longest common subsequence~(LCS)} problem is fundamental and well studied in computer science.
The most common application of the LCS problem is measuring similarity between strings,
which can be used in many applications such as the \texttt{diff} tool, the time series data analysis~\cite{ref:journal/WASJ/Khan13}, and in bioinformatics.

One of the major disadvantages of LCS as a measure of similarity is that
LCS cannot consider consecutively matching characters effectively.
For example, for strings $X = \mathtt{ATGG}, Y = \mathtt{ATCGGC}$ and $Z = \mathtt{ACCCTCCCGCCCG}$,
$\mathtt{ATGG}$ is the LCS of $X$ and $Y$, which is also the LCS of $X$ and $Z$.
Benson \emph{et al.}~\cite{ref:conf/sisap/Benson13} introduced the \emph{longest common subsequence in $k$ length substrings~(LCS$_k$)} problem, where the subsequence needs to be a concatenation of $k$ length substrings of given strings.
For example, for strings $X = \mathtt{ATCTATAT}$ and $Y = \mathtt{TAATATCC}$, $\mathtt{TAAT}$ is an LCS$_2$
since $X[4:5] = Y[1:2] = \mathtt{TA}$ and $X[7:8] = Y[5:7] = \mathtt{AT}$, and no longer one exists.
They showed a quadratic time algorithm for it, and
Deorowicz and Grabowski~\cite{ref:journals/ipl/Deorowicz14} proposed several algorithms,
such as a quadratic worst-case time algorithm for unbounded $k$ and a fast algorithm on average.

Paveti{\'c} \emph{et al.}~\cite{ref:preprint/Pavetic14} considered 
the \emph{longest common subsequence in at least $k$ length substrings~(\problemabbr)} problem,
where the subsequence needs to be a concatenation of \emph{at least} $k$ length substrings of given strings.
They argued that {\problemabbr} would be more appropriate than LCS$_k$ as a similarity measure of strings.
For strings $X = \mathtt{ATTCGTATCG}$, $Y = \mathtt{ATTGCTATGC}$, and $Z = \mathtt{AATCCCTCAA}$, 
$\mathit{LCS}_2(X, Y) = \mathit{LCS}_2(X, Z) = 4$, where $\mathit{LCS}_2(A, B)$ denotes the length of an LCS$_2$ between $A$ and $B$.
However, it seems that $X$ and $Y$ are more similar than $X$ and $Z$.
Instead, if we consider {\problemabbrk{2}}, we have 
$\mathit{LCS}_{2^+}(X, Y) = 6 > 4 = \mathit{LCS}_{2^+}(X, Z)$,
that better fits our intuition.
The notion of {\problemabbr} is applied to bioinformatics~\cite{ref:journal/natcommun/Sovic16}.

Paveti{\'c} \emph{et al.} showed that {\problemabbr} can be computed in $O(m + n + r \log r + r \log n)$ time,
where $m, n$ are lengths of the input strings and $r$ is the total number of matching $k$ length substring pairs between the input strings.
Their algorithm is fast on average, but in the worst case, the running time is $O(mn \log (mn))$.
Independently, Benson~\emph{et al.}~\cite{ref:conf/sisap/Benson13} proposed an $O(kmn)$ worst-case time algorithm
for the {\problemabbr} problem.

In this paper, we first propose an algorithm to compute {\problemabbr} in $O(mn)$ worst-case time by a simple dynamic programming. 
Secondly, we introduce the \emph{\opproblem~(\opproblemabbr)} problem.
Order-isomorphism is a notion of equality of two numerical strings,
intensively studied in the \emph{order-preserving matching} problem\footnote{
	Since the problem is motivated by the order-preserving matching problem,
	we abbreviate it to the {\opproblemabbr} problem.
}~\cite{ref:journal/TCS/Kim14,ref:journals/ipl/Kubica13}.
{\opproblemabbr} is a natural definition of similarity between numerical strings,
and can be used in time series data analysis.
The {\opproblemabbr} problem cannot be solved as simply as the {\problemabbr} problem
due to the properties of the order-isomorphism.
However, we will show that the {\opproblemabbr} problem can also be solved in $O(mn)$ worst-case time by an easy-to-implement algorithm, which is one of the main contributions of this paper.
Finally, we report experimental results.

\section{Preliminaries}
We assume that all strings are over an \emph{alphabet} $\Sigma$.
The length of a string $X = (X[1], X[2], \cdots, X[n]) $ is denoted by $\len{X} = n$.
A $substring$ of $X$ beginning at $i$ and ending at $j$ is denoted by $\substr{X}{i}{j} = (X[i], X[i+1], \cdots, X[j-1], X[j])$.
We denote $\sublen{X}{i}{l} = \substr{X}{i}{i+l-1}$ and $\sublastlen{X}{j}{l} = \substr{X}{j-l+1}{j}$. 
Thus $\sublen{X}{i}{l} = \sublastlen{X}{i+l-1}{l}$.
We write $\prefix{X}{i}$ and $\suffix{X}{j}$ to denote the \emph{prefix} $\substr{X}{1}{i}$ and the \emph{suffix} $\substr{X}{j}{n}$ of $X$, respectively.
Note that $\prefix{X}{0}$ is the empty string.
The reverse of a string $X$ is denoted by $\reverse{X}$, and
the operator $\cdot$ denotes the concatenation.
We simply denote a string $X = (X[1], X[2], \cdots, X[n])$ as $X=X[1]X[2]\cdots X[n]$ when clear from the context.

We formally define the {\problemabbr} problem as follows.
\begin{definition}[{\problemabbr} problem~\cite{ref:conf/sisap/Benson13,ref:preprint/Pavetic14}\footnote{
		The formal definition given by Paveti{\'c} \emph{et al.}~\cite{ref:preprint/Pavetic14} contains a minor error,
		i.e., they do not require that each chunk is identical, while Benson~\emph{et al.}~\cite{ref:conf/sisap/Benson13} and we do (confirmed by F. Paveti{\'c}, personal communication, October 2016).
	}] \label{def:LCSkStandard}
	Given two strings $X$ and $Y$ of length $m$ and $n$, respectively, and an integer $k \ge 1$,
	we say that $Z$
	is a \emph{common subsequence in at least $k$ length substrings} of $X$ and $Y$, if there exist $i_1, \cdots, i_t$ and $j_1, \cdots, j_t$ such that
	$\sublen{X}{i_s}{l_s} = \sublen{Y}{j_s}{l_s} = \sublen{Z}{p_s}{l_s}$ and
	$l_s \ge k$ for $1 \le s \le t$, 
	and $i_{s} + l_{s} \le i_{s+1}$, $j_{s} + l_{s} \le j_{s+1}$
	and $p_{s+1} = p_s + l_s$ for $1 \le s < t$, $p_1 = 1$ and $|Z| = p_t + l_t - 1$.
	The \emph{longest common subsequence in at least $k$ length substrings~(\problemabbr)} problem asks for the length of an {\problemabbr} of $X$ and $Y$. 
\end{definition}
Remark that the {\problemabbrk{1}} problem is equivalent to the standard LCS problem.
Without loss of generality, we assume that $n \ge m$ through the paper.

\begin{example}
	For strings $X = \mathtt{acdbacbc}$ and $Y = \mathtt{aacdabca}$,
	$Z = \mathtt{acdbc}$ is the {\problemabbrk{2}} of $X$ and $Y$, since $\sublen{X}{1}{3} = \sublen{Y}{2}{3} = \mathtt{acd} = \sublen{Z}{1}{3}$
	and $\sublen{X}{7}{2} = \sublen{Y}{6}{2} = \mathtt{bc} = \sublen{Z}{4}{2}$.
	Note that the standard LCS of $X$ and $Y$ is $\mathtt{acdabc}$.
\end{example}

The main topic of this paper is to give an efficient algorithm for computing the longest common subsequence \emph{under order-isomorphism}, defined below.

\begin{definition}[Order-isomorphism~\cite{ref:journal/TCS/Kim14,ref:journals/ipl/Kubica13}]
	Two strings $S$ and $T$ of the same length $l$ over an ordered alphabet are \emph{order-isomorphic} if
	$S[i] \leq S[j] \ \Longleftrightarrow  \ T[i] \leq T[j]$
	for any $1 \leq i,j \leq l$.
	We write $S \orderisomorphic T$ if $S$ is order-isomorphic to $T$, and
	$S \not\orderisomorphic T$ otherwise.
\end{definition}

\begin{example}
	For strings $S = (32, 40, 4, 16, 27)$, $T = (28, 32, 12, 20, 25)$ and $U = (33, 51, 10,$ $22, 42)$,
	we have $S \orderisomorphic T$, $S \not\orderisomorphic U$, and $T \not\orderisomorphic U$.
\end{example}

\begin{definition}[{\opproblemabbr} problem] \label{def:opLCSproblem}
	The \emph{\opproblemabbr} problem is 
	defined as the problem
	obtained from Definition~\ref{def:LCSkStandard} by replacing 
	the matching relation $\sublen{X}{i_s}{l_s} = \sublen{Y}{j_s}{l_s} = \sublen{Z}{p_s}{l_s}$ with order-isomorphism
	$\sublen{X}{i_s}{l_s} \orderisomorphic \sublen{Y}{j_s}{l_s} \orderisomorphic \sublen{Z}{p_s}{l_s}$.
\end{definition}

\begin{example}\label{example:op-LCSk+}
	For strings $X = (14, 84, 82, 31, 74, 68, 87, 11, 20, 32)$ and $Y = (21, 64,$ $2, 83, 73, 51, 5, 29, 7, 71)$,
	$Z = (1, 3, 2, 31, 74, 68, 87)$ is an {\opproblemabbrk{3}} of $X$ and $Y$ since
	$\sublen{X}{1}{3} \orderisomorphic \sublen{Y}{3}{3} \orderisomorphic \sublen{Z}{1}{3} $ 
	and $\sublen{X}{4}{4} \orderisomorphic \sublen{Y}{7}{4} \orderisomorphic \sublen{Z}{4}{4}$.
\end{example}

The {\opproblemabbr} problem does not require that
$( \sublen{X}{i_1}{l_1} \cdot \sublen{X}{i_2}{l_2} \cdot \; \cdots \; \cdot \sublen{X}{i_t}{l_t}) \approx 
(\sublen{Y}{j_1}{l_1} \cdot \sublen{Y}{j_2}{l_2} \cdot \; \cdots \; \cdot \sublen{Y}{j_t}{l_t})$.
Therefore, the {\opproblemabbrk{1}} problem makes no sense.
Note that the {\opproblemabbr} problem with this restriction
is \textbf{NP}-hard already for $k=1$~\cite{ref:conf/cpm/Bouvel07}.
\section{The {\problemabbr} Problem}\label{sec:standard}
In this section, we show that the {\problemabbr} problem
can be solved in $O(mn)$ time by dynamic programming.
We define $\match{i}{j}{l} = 1$ if $\sublastlen{X}{i}{l} = \sublastlen{Y}{j}{l}$, and $0$ otherwise.
Let $\cop{i}{j}$ be the length of an {\problemabbr} of $\prefix{X}{i}$ and $\prefix{Y}{j}$,
and $\A = \left\{\cop{i - l}{j - l} + l \cdot \match{i}{j}{l} : k \le l \le \min\{i, j\}  \right\}$.
Our algorithm is based on the following lemma.

\begin{lemma}[\cite{ref:conf/sisap/Benson13}]
	\label{lemma:optimal-substructure}
	For any $k \leq i \leq m$ and $k \leq j \leq n$, 
	\begin{align} \label{eq:C}
		\cop{i}{j}=
		\max\left(\{\cop{i}{j-1}, \cop{i-1}{j}\} \cup \A \right),		
	\end{align}
	and $\cop{i}{j} = 0$ otherwise.
\end{lemma}

The naive dynamic programming algorithm based on Equation~(\ref{eq:C}) takes $O(m^2n)$ time,
because  for each $i$ and $j$,
the naive algorithm for computing $\max\A$ takes $O(m)$ time assuming $n \ge m$.
Therefore, we focus on how to compute
$\max \A$ in constant time
for each $i$ and $j$
in order to solve the problem in $O(mn)$ time.
It is clear that if $\match{i}{j}{l_1} = 0$ then $\match{i}{j}{l_2} = 0$ for all valid $l_2 \ge l_1$,
and $\cop{i'}{j'} \ge \cop{i' - l'}{j'-l'}$ for all valid $i', j'$ and $l' > 0$.
Therefore, in order to compute 
$\max \A$,
it suffices to compute $\max_{k \le l \le L[i, j]}\{\cop{i - l}{j - l} + l \}$,
where  $L[i, j] = \max\{l: \sublastlen{X}{i}{l} = \sublastlen{Y}{j}{l}\}$.

We can compute $L[i, j]$ for all $0 \le i \le m$ and $0 \le j \le n$ in $O(mn)$ time by dynamic programming
because the following equation clearly holds:
\begin{equation} \label{eq:DP-LCE}
	L[i, j] = 
	\begin{cases}
		L[i-1, j-1] + 1 & \text{(if $i, j > 0$ and $X[i]=Y[j]$)} \\
		0 & \text{(otherwise)}.
	\end{cases}
\end{equation}

Next, we show how to compute $\max_{k \le l \le L[i, j]}\{\cop{i - l}{j - l} + l \}$ in constant time for each $i$ and $j$.
Assume that the table $L$ has already been computed.
Let $\cmaxijk{i}{j}{k} = \max_{ k \le l \le L[i, j]}\{\cop{i-l}{j-l} + l \}$ if $L[i, j] \ge k$, and $-1$ otherwise.

\begin{lemma} \label{lemma:cmax2} 
	For any $0 \leq i \leq m$ and $0 \leq j \leq n$,  
	if $L[i, j] > k$ then $\cmaxijk{i}{j}{k} = \max\{\cmaxijk{i-1}{j-1}{} + 1, \cop{i-k}{j-k} + k\}$.
\end{lemma}
\begin{proof}
	Let $l = L[i, j]$.
	Since $L[i, j] > k$, we have $L[i-1, j-1] = l - 1 \ge k$, and $\cmaxijk{i-1}{j-1}{k} \neq -1$.
	Therefore, 
	$	\cmaxijk{i-1}{j-1}{k} = \max_{k \le l' \le l-1}\{\cop{i- 1 - l'}{j- 1 - l'} + l'\} = \max_{k+1 \le l' \le l}\{\cop{i - l'}{j - l'} + l'\} - 1.$
	Hence, $\cmaxijk{i}{j}{k} = \max_{ k \le l' \le l}\{\cop{i-l'}{i-l'} + l'\} = \max\{\cmaxijk{i-1}{j-1}{}+ 1, \cop{i-k}{j-k} + k\}$.
	\qed
\end{proof}
By Lemma \ref{lemma:cmax2} and the definition of $\cmaxijk{i}{j}{}$, we have
\begin{equation} \label{eq:cmax}
	\cmaxijk{i}{j}{k} =
	\begin{cases}
		\max\{\cmaxijk{i-1}{j-1}{}+1, \cop{i-k}{j-k}+k \} & \text{(if $L[i, j]>k$)} \\
		\cop{i-k}{j-k}+k & \text{(if $L[i, j]=k$)} \\
		-1 & \text{(otherwise).}
	\end{cases}
\end{equation}
Equation~(\ref{eq:cmax}) shows that each $\cmaxijk{i}{j}{}$ can be computed in constant time if $L[i,j]$, $M[i-1,j-1]$, and $C[i-k,j-k]$ have already been computed.

\begin{figure}[t]
	\begin{minipage}[t]{0.48\linewidth}
		\includegraphics[width=5.5cm]{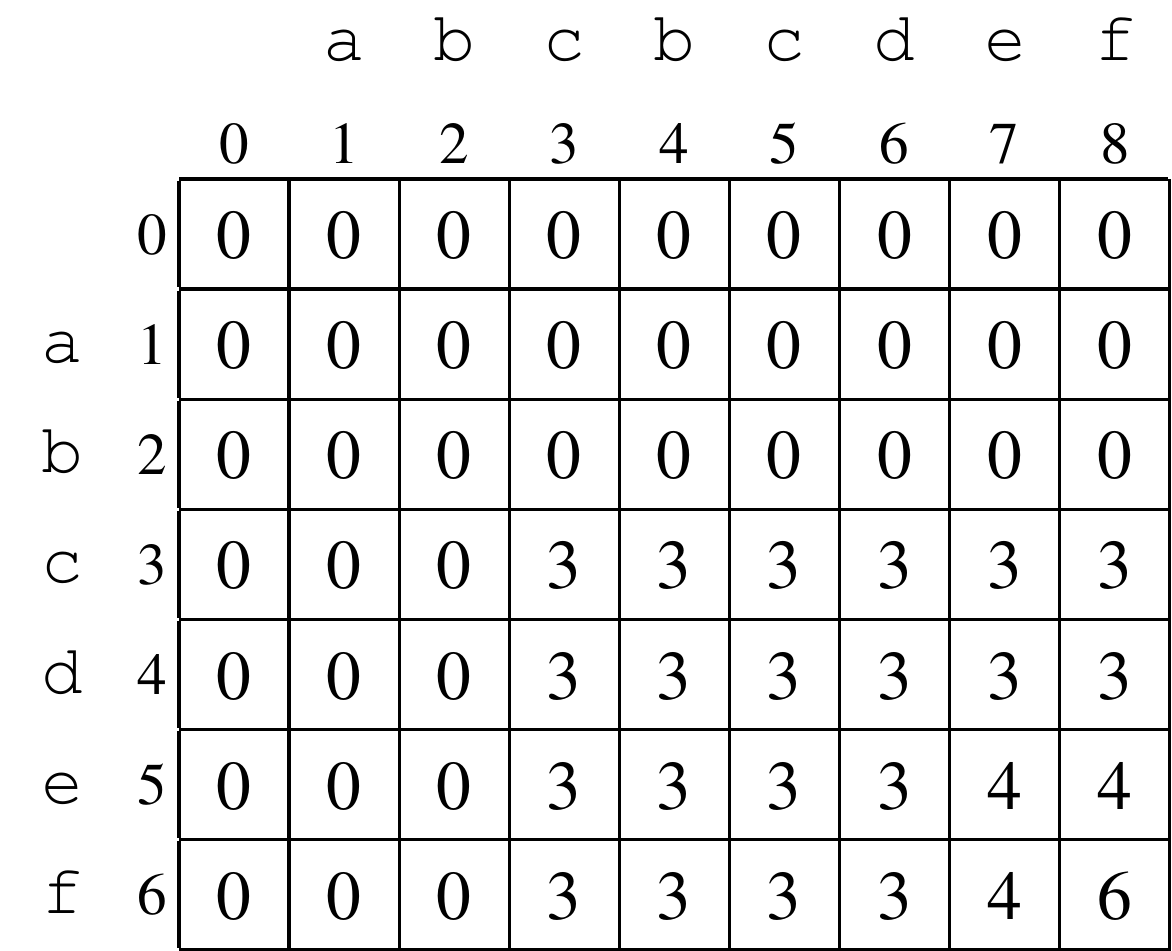}
		\subcaption{Table $C$ for the {\problemabbrk{3}} problem}
		\label{fig:standard/example}
	\end{minipage}
	\begin{minipage}[t]{0.48\linewidth}
		\includegraphics[width=5.5cm]{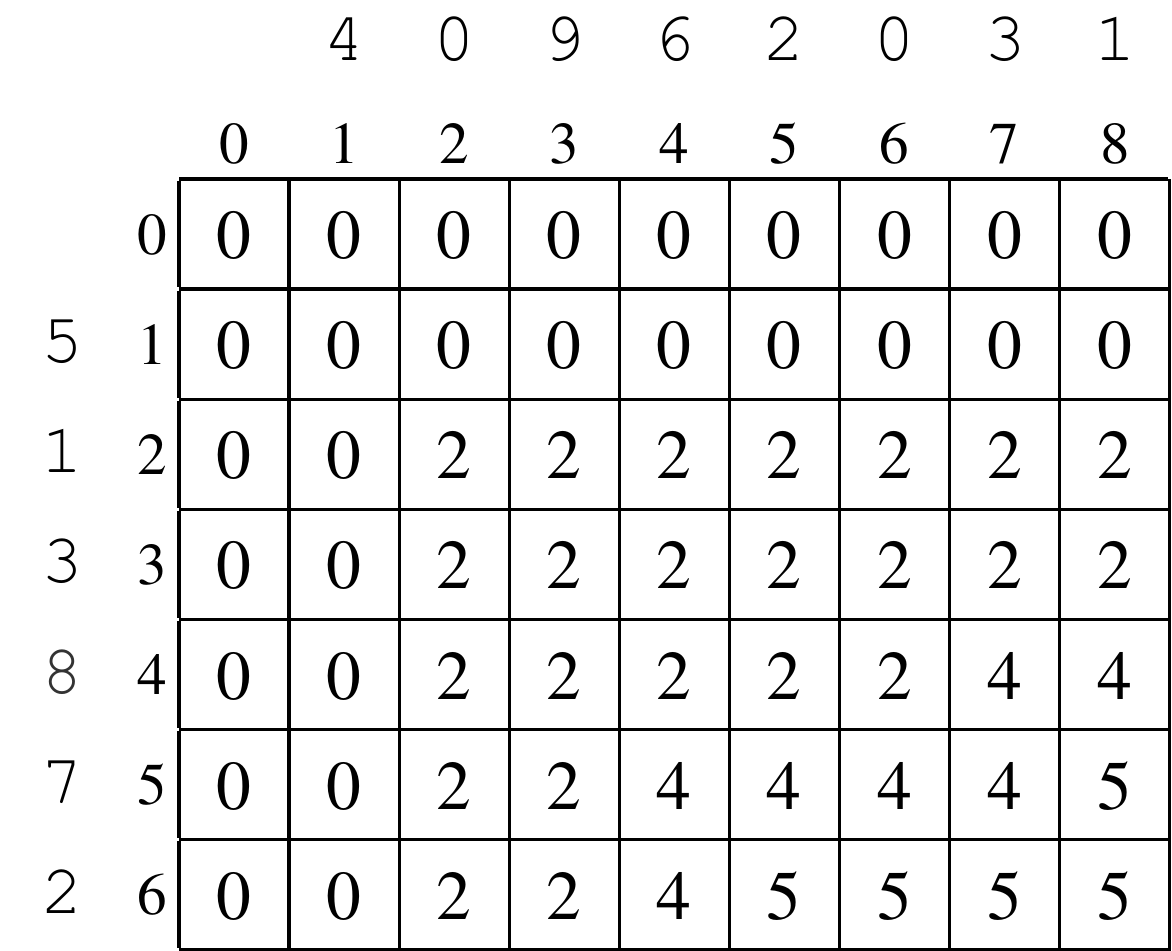}
		\subcaption{Table $C$ for the {\opproblemabbrk{2}} problem}
		\label{fig:OPLCS/example}
	\end{minipage}
	\caption{Examples of computing {\problemabbrk{3}} and {\opproblemabbrk{2}}}\label{fig:examples}
\end{figure}

We can fill in tables $\co$, $L$ and $M$ of size $(m+1) \times (n+1)$ based on Equations~(\ref{eq:C}), (\ref{eq:DP-LCE}) and (\ref{eq:cmax}) in $O(mn)$ time
by dynamic programming.
An example of computing {\problemabbrk{3}}
is shown in Fig.~\ref{fig:examples}\subref{fig:standard/example}.
We note that {\problemabbr} itself (not only its length) can be extracted from the table $\co$ in $O(m+n)$ time,
by tracing back in the same way as the standard dynamic programming algorithm for the standard LCS problem.
Our algorithm requires $O(mn)$ space since we use three tables of size $(m+1) \times (n+1)$.
Note that if we want to compute only the length of an {\problemabbr}, 
the space complexity can be easily reduced to $O(km)$.
Hence, we get the following theorem.

\begin{theorem}
	The {\problemabbr} problem can be solved in $O(mn)$ time and $O(km)$ space.
\end{theorem}
\section{The {\opproblemabbr} Problem}
In this section, we show that the {\opproblemabbr} problem can be solved in $O(mn)$ time
as well as the {\problemabbr} problem.
We redefine $\cop{i}{j}$ to be the length of an {\opproblemabbr} of $\prefix{X}{i}$ and $\prefix{Y}{j}$,
and $\match{i}{j}{l} = 1$ if $\sublastlen{X}{i}{l} \approx \sublastlen{Y}{j}{l}$, and $0$ otherwise.
It is easy to prove that Equation~(\ref{eq:C}) also holds with respect to the order-isomorphism.
However, the {\opproblemabbr} problem cannot be solved as simply as
the {\problemabbr} problem
because Equations~(\ref{eq:DP-LCE}) and (\ref{eq:cmax}) do not hold with respect to the order-isomorphism, as follows.
For two strings $A, B$ of length $l$ such that $A \orderisomorphic B$,
and two characters $a, b$ such that $A \cdot a \not\orderisomorphic B \cdot b$,
the statement ``$\suffix{(A \cdot a)}{i} \not\orderisomorphic \suffix{(B \cdot b)}{i}$ for all $1 \le i \le l+1$''
is not always true.
For example, for strings $A = (32, 40, 4, 16, 27)$, $B = (28, 32, 12, 20, 25)$, $A' = A \cdot (41)$ and $B' = B \cdot (26)$,
we have $A \orderisomorphic B$, $A' \not\orderisomorphic B'$, and $\suffix{A'}{3} \orderisomorphic \suffix{B'}{3}$.
Moreover, for $A'' = A \cdot (15)$ and $B'' = B \cdot (22)$, we have $\suffix{A''}{5} \orderisomorphic \suffix{B''}{5}$.
These examples show that Equations~(\ref{eq:DP-LCE}) and (\ref{eq:cmax}) do not hold with respect to the order-isomorphism.
Therefore, we must find another way to compute
$\max_{k \le l' \le l}\{\cop{i - l'}{j - l'} + l' \}$,
where
 $l = \max\{l': \sublastlen{X}{i}{l'} \approx \sublastlen{Y}{j}{l'}\}$
in constant time.

First, we consider how to find  $\max\{l: \sublastlen{X}{i}{l} \orderisomorphic \sublastlen{Y}{j}{l} \} $
in constant time.
We define the \emph{order-preserving longest common extension~(op-LCE)} query on strings $S_1$ and $S_2$ as follows.
\begin{definition}[op-LCE query]
	Given a pair $(S_1,S_2)$ of strings,
	an \emph{op-LCE query} is a pair of indices $i_1$ and $i_2$ of $S_1$ and $S_2$, respectively,
	which asks $\opLCEps{i_1}{i_2}{S_1}{S_2}=\max\{l: \sublen{S_1}{i_1}{l} \orderisomorphic \sublen{S_2}{i_2}{l}\}$.
\end{definition}

Since $\max\{l: \sublastlen{X}{i}{l} \orderisomorphic \sublastlen{Y}{j}{l} \} = 
\opLCEps{m-i+1}{n-j+1}{\reverse{X}}{\reverse{Y}} $,
we can find $\max\{l: \sublastlen{X}{i}{l} \orderisomorphic \sublastlen{Y}{j}{l} \}$
by using op-LCE queries on $\reverse{X}$ and $\reverse{Y}$.
Therefore, we focus on how to answer op-LCE queries on $S_1$ and $S_2$ in constant time with at most $O(|S_1||S_2|)$ time preprocessing.
Hereafter we write $\opLCE{i_1}{i_2}$ for $\opLCEps{i_1}{i_2}{S_1}{S_2}$ fixing two strings $S_1$ and $S_2$. 

If $S_1$ and $S_2$ are strings over a polynomially-bounded integer alphabet $\{1, \cdots, (|S_1| + |S_2|)^c \}$ for an integer constant $c$,
op-LCE queries can be answered in $O(1)$ time and $O(|S_1| + |S_2|)$ space
with $O((|S_1| + |S_2|) \log^2\log(|S_1| + |S_2|)/\log\log\log(|S_1|+|S_2|))$ time preprocessing,
by using
 the \emph{incomplete generalized op-suffix-tree}~\cite{ref:journal/TCS/Crochemore15} of $S_1$ and $S_2$
and finding the \emph{lowest common ancestor~(LCA)}~\cite{Bender2000} in the op-suffix-tree.
The proof is similar to that for LCE queries in the standard setting~\cite{ref:book/Gusfield97}.

However, implementing the incomplete generalized op-suffix-tree is quite difficult.
Therefore, we introduce another much simpler method to answer op-LCE queries in $O(1)$ time with $O(|S_1||S_2|)$ time preprocessing.
In a preprocessing step, our algorithm fills in the table $\opLCE{i_1}{i_2}$
for all $1 \le i_1 \le \len{S_1}$ and $1 \le i_2 \le \len{S_2}$ in $O(|S_1||S_2|)$ time.
Then, we can answer op-LCE queries in constant time.

In the preprocessing step, we use the \emph{$Z$-algorithm}~\cite{ref:book/Gusfield97,ref:journal/PRL/Hasan15}
that calculates the following table efficiently.
\begin{definition}[$Z$-table]
	The \emph{$Z$-table} $\Z{S}$ of a string $S$ is defined by
	$\Zi{S}{i} = \max\{l: \sublen{S}{1}{l} \orderisomorphic \sublen{S}{i}{l}\}$ for each $1 \le i \le |S|$.
\end{definition}
By definition, we have 
\begin{equation}
\label{eq:opLCE-Z}
\opLCE{i_1}{i_2} = \min\bigl\{\Zi{\suffix{\left(S_1 \cdot S_2\right)}{i_1}}{|S_1| - i_1 + i_2 + 1}, \ |S_1| - i_1 + 1 \bigr\}.
\end{equation}

If we use the $Z$-algorithm and Equation~(\ref{eq:opLCE-Z}) naively, it
takes $O((|S_1|+|S_2|)^2\log(|S_1|+|S_2|))$ time
to compute $\opLCE{i_1}{i_2}$ for all $1 \le i_1 \le |S_1|$ and $1 \le i_2 \le |S_2|$,
because the $Z$-algorithm requires $O(|S|\log|S|)$ time to compute $\Z{S}$ for a string $S$.
We extend the $Z$-algorithm to compute $\Z{\suffix{S}{i}}$ for \emph{all} $1 \le i \le |S|$ totally in $O(|S|^2)$ time.

In order to verify the order-isomorphism in constant time with preprocessing, 
Hasan \textit{et al.}~\cite{ref:journal/PRL/Hasan15} used tables called $\Prev{S}$ and $\Next{S}$.
For a string $S$ where all the characters are distinct\footnotemark[5],
$\Prev{S}$ and $\Next{S}$ are defined as 
\begin{align*}
	&\text{$\Prev{S}[i] = j$ if $\text{there exists } j = \argmax_{1 \le k < i}\{S[k]: S[k] < S[i] \}$, and}& \text{$-\infty$ otherwise}  \\
	&\text{$\Next{S}[i] = j$ if $\text{there exists } j = \argmin_{1 \le k < i}\{S[k]: S[k] > S[i] \}$, and}& \text{$\infty$ otherwise}
\end{align*}
for all $1 \le i \le |S|$.
Their algorithm requires $O(|S|\log|S|)$ time to compute the tables $\Prev{S}$ and $\Next{S}$,
and all operations except computing the tables take only $O(|S|)$ time.
Therefore, if we can compute tables $\Prev{\suffix{S}{i}}$ and $\Next{\suffix{S}{i}}$
for each $1 \le i \le |S|$ in $O(|S|)$ time with $O(|S|\log|S|)$ time preprocessing,
$\Z{\suffix{S}{i}}$ for all $1 \le i \le |S|$ can be computed in $O(|S|^2)$ time.
We also assume that all the characters in $S$ are distinct\footnotemark[5].

\footnotetext[5]{
	Hasan \textit{et al.}~\cite{ref:journal/PRL/Hasan15} assume that characters in a string are distinct.
	If the assumption is false, use Lemma~4 in \cite{ref:journal/IPL/Cho15} in order to verify the order-isomorphism, 
	that is, modify line~10 of Algorithm~4 in \cite{ref:journal/PRL/Hasan15} and line~\ref{alg_line:prev} and \ref{alg_line:next} in Algorithm~\ref{alg:opLCE}.
	Note that $\mathit{Prev}$ and $\mathit{Next}$ are denoted as $\mathit{LMax}$ and $\mathit{LMin}$ in \cite{ref:journal/IPL/Cho15}, respectively,
	with slight differences.
}

In order to compute the tables $\Prev{\suffix{S}{i}}$ and $\Next{\suffix{S}{i}}$, we modify a sort-based algorithm presented in Lemma~1 in \cite{ref:journals/ipl/Kubica13}
instead of the algorithm in \cite{ref:journal/PRL/Hasan15} that uses a balanced binary search tree.
First, for computing $\Prev{S}$ (resp.\ $\Next{S}$),
we stably sort positions of $S$ with respect to their elements in ascending (resp.\ descending) order.
We can compute $\Prev{\suffix{S}{i}}$ and $\Next{\suffix{S}{i}}$ for each $1 \le i \le \len{S}$ in $O(|S|)$ time
by using the sorted tables and the stack-based algorithm presented in \cite{ref:journals/ipl/Kubica13},
ignoring all elements of the sorted tables less than $i$.

\begin{algorithm2e}[p]
\caption{The algorithm for computing op-LCE queries}
\label{alg:opLCE}

\SetKwFunction{FPreprocessing}{preprocess}
\SetKwProg{Fn}{Function}{}{}
\Fn{\FPreprocessing{$S$, $i$, $S'$, $S''$}}{
	Let $s$ and $t$ be empty stacks that support $\mathtt{push}$, $\mathtt{top}$, and $\mathtt{pop}$ operations\;
	Let $\mathit{Prev}$ and $\mathit{Next}$ be tables of size $|S| - i + 1$\;
	\For{$j \leftarrow 1$ \KwTo $|S|$}{
		\If{$S'[j] \ge i$}{
			\lWhile{$s \neq \emptyset$ and $\stTop{s}> S'[j]$}{$\stPop{s}$}
			\lIf{$s = \emptyset$}{$\mathit{Prev}[S'[j] - i + 1] \leftarrow -\infty$}\label{alg_line:prev}
			\lElse{$\mathit{Prev}[S'[j] - i + 1] \leftarrow \stTop{s} - i + 1$}			
			$\stPush{s}{S'[j]}$\;				
		}
		\If{$S''[j] \ge i$}{
			\lWhile{$t \neq \emptyset$ and $\stTop{t}> S''[j]$}{$\stPop{t}$}
			\lIf{$t = \emptyset$}{$\mathit{Next}[S''[j] - i + 1] \leftarrow \infty$}\label{alg_line:next}
			\lElse{$\mathit{Next}[S''[j] - i + 1] \leftarrow \stTop{t} - i + 1$}			
			$\stPush{t}{S''[j]}$\;				
		}
	}
	\KwRet{$(\mathit{Prev}, \mathit{Next})$}\;
}

\SetKwFunction{FZ}{Z-function}
\SetKwProg{Fn}{Function}{}{}

\Fn{\FZ{$S, i_1, S', S''$}}{
	$(\mathit{Prev}, \mathit{Next}) \leftarrow \mathtt{preprocess}(S, i_1, S', S'')$;
	$S \leftarrow \suffix{S}{i_1}$\;
	Do the same operations described in line 3-17 of Algorithm~4 in \cite{ref:journal/PRL/Hasan15}\;
	\KwRet{$Z$}\;
}

\SetKwFunction{FOPLCE}{preprocess-opLCE}
\SetKwProg{Fn}{Function}{}{}
	\Fn{\FOPLCE{$S_1, S_2$}}{
	
	Let $\mathit{opLCE}$ be a table of size $|S_1| \times |S_2|$;
	$S \leftarrow S_1 \cdot S_2$\;
	Let $S'$ and $S''$  be stably sorted positions of $S$ with respect to their elements in ascending and descending order, respectively\;
	
	\For{$i_1 \leftarrow 1$ \KwTo $|S_1|$}{
		$Z \leftarrow \mathtt{Z\mathchar`-function}(S, i_1, S', S'')$\; \label{alg_line:Z-function}
		\For{$i_2 \leftarrow 1$ \KwTo $|S_2|$}{
			$\opLCE{i_1}{i_2} \leftarrow \min\Bigl\{Z[|S_1| - i_1 + i_2 + 1], |S_1| - i_1 + 1 \Bigr\}$\;
		}
	}	
	\KwRet{$\mathit{opLCE}$}\;
}

\end{algorithm2e}

Algorithm~\ref{alg:opLCE} shows the pseudocode of the op-LCE algorithm based on the $Z$-algorithm.
The $\mathtt{push}(x)$ operation inserts $x$ on the top of the stack, 
$\mathtt{top}()$ returns the top element in the stack,
and $\mathtt{pop}()$ removes it.
Algorithm~\ref{alg:opLCE} takes $O(\len{S_1}\len{S_2})$ time as discussed above.
The total space complexity is $O(\len{S_1}\len{S_2})$
because the $Z$-algorithm requires linear space~\cite{ref:journal/PRL/Hasan15}, and
the table $\opL$  needs $O(|S_1||S_2|)$ space.
Hence, we have the following lemma.
\begin{lemma}\label{lemma:op-LCE-simple}
	op-LCE queries on $S_1$ and $S_2$ can be answered in $O(1)$ time and $O(|S_1||S_2|)$ space with $O(|S_1||S_2|)$ time preprocessing.
\end{lemma}

\begin{algorithm2e}[t]
\caption{The algorithm for the {\opproblemabbr} problem}
\label{alg:OPLCS}
\KwIn{A string $X$ of length $m$, a string $Y$ of length $n$, and an integer $k$}
\KwOut{The length of an {\opproblemabbr} between $X$ and $Y$}

Let $\co$ be a table of size $(m+1) \times (n+1)$ initialized by $0$\;
Let $R_i$ for $-n+k \le i \le m-k$ be semi-dynamic RMQ data structures\;

$\opL \leftarrow \mathtt{preprocess\mathchar`-opLCE}(\reverse{X}, \reverse{Y})$\;
\For{$i \leftarrow 0$ \KwTo $m-k$} {
	\lIf{$i < k$}{$n' \leftarrow n-k$}
	\lElse{$n' \leftarrow k-1$}	
	\lFor{$j \leftarrow 0$ \KwTo $n'$}{
		$R_{i-j}.\mathtt{prepend}(\cop{i}{j} - \min\{i, j\})$%
	}
}

\For{$i \leftarrow k$ \KwTo $m$}{ 
    \For{$j \leftarrow k$ \KwTo $n$}{
		$l \leftarrow \opLCE{m -i + 1}{n - j + 1}$\;
		\lIf{$l \ge k $}{		
			$M \leftarrow R_{i-j}.\mathtt{rmq}(k, l) + \min\{i, j \}$%
		}
		\lElse{
			$M \leftarrow 0$%
		}

        $\cop{i}{j} \leftarrow \max\{\cop{i}{j-1}, \cop{i-1}{j}, M \}$\;
        $R_{i-j}.\mathtt{prepend}(\cop{i}{j} - \min\{i, j\})$\;
    }
}
\KwRet{$\cop{m}{n}$}\;
\end{algorithm2e}

Let $\opLq{i}{j}$ be the answer to the op-LCE query on $\reverse{X}$ and $\reverse{Y}$ with respect to the index pair $(i, j)$.
We consider how to find the maximum value of $\cop{i-l}{j-l} + l$ for $k \le l\le \opLq{m-i+1}{n-j+1}$ in constant time.
We use a \emph{semi-dynamic range maximum query~(RMQ)} data structure that maintains a table $A$ and supports the following two operations:
\begin{itemize}
	\item[\small$\bullet$] $\mathtt{prepend}(x)$: add $x$ to the beginning of $A$ in $O(1)$ amortized time.
	\item[\small$\bullet$] $\mathtt{rmq}(i_1, i_2)$: return the maximum value of $\substr{A}{i_1}{i_2}$ in $O(1)$ time.
\end{itemize}
The details of the semi-dynamic RMQ data structure will be given in Section~\ref{sec:rmq}.

By using the semi-dynamic RMQ data structures and the following obvious lemma, we can find
$\max_{k \le l \le \opLq{m-i+1}{n-j+1}}\{\cop{i-l}{j-l} + l\}$
for all $1 \le i \le m$ and $1 \le j \le n$ in totally $O(mn)$ time.
\begin{lemma} \label{lem:RMQ}
	We may assume that $i \ge j$ without loss of generality.
	Let $A[l] = C[i-l, j-l] + l$ and $A'[l] = C[i-l, j-l]-j+l$ for each $1 \le l \le j$.	
	For any $1 \le i_1, i_2 \le |A|$, we have $\max_{i_1 \le l \le i_2}{A[l]}  = (\max_{i_1 \le l \le i_2}{A'[l]}) + j$
	and $\argmax_{i_1 \le l \le i_2}{A[l]}  = \argmax_{i_1 \le l \le i_2}{A'[l]}$.
\end{lemma}

Algorithm~\ref{alg:OPLCS} shows our algorithm to compute {\opproblemabbr}.
An example of computing {\opproblemabbrk{2}}
is shown in Fig.~\ref{fig:examples}\subref{fig:OPLCS/example}.
As discussed above, the algorithm runs in $O(mn)$ time.
Each semi-dynamic RMQ data structure requires linear space
and a total of $O(mn)$ elements are maintained by the semi-dynamic RMQ data structures.
Therefore, the total space of semi-dynamic RMQ data structures is $O(mn)$.
Consequently, the total space complexity is $O(mn)$. Hence, we have the following theorem.
\begin{theorem}
	The {\opproblemabbr} problem can be solved in $O(mn)$ time and space.
\end{theorem}

\section{The Semi-dynamic Range Minimum/Maximum Query}
\label{sec:rmq}

\begin{figure}[t]
	\centering
	\includegraphics[scale=0.6]{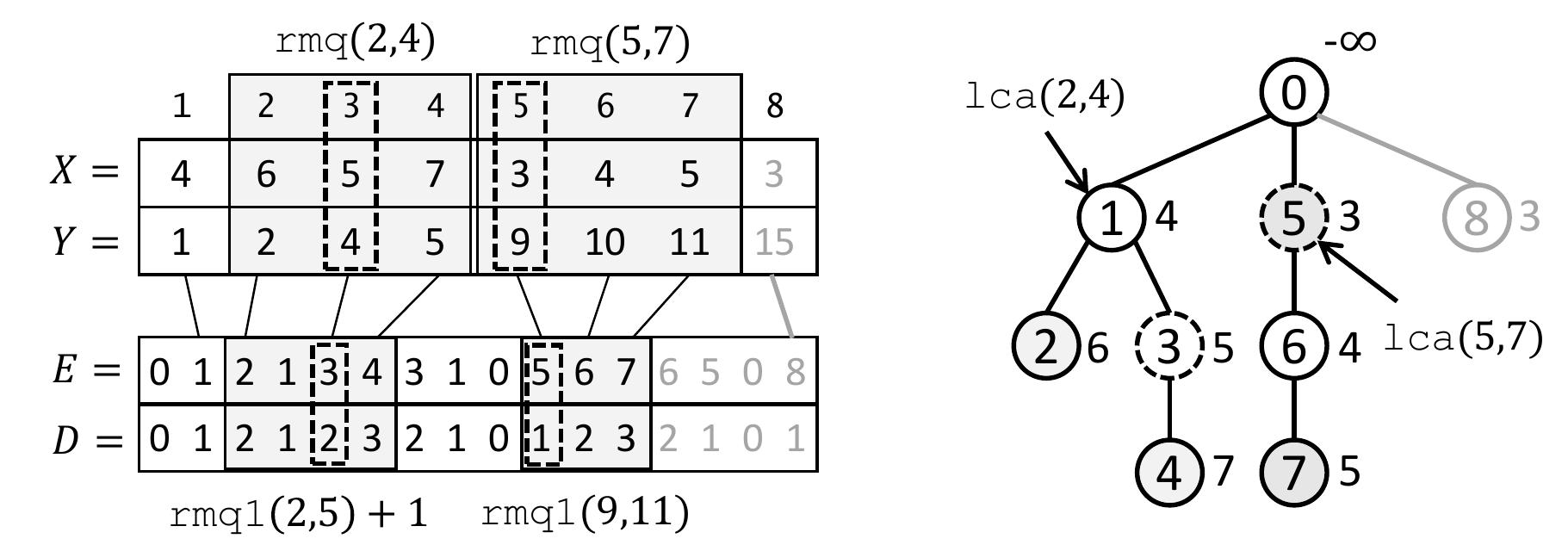}\\
	\caption{An example of searching for the RMQ by using a 2d-Min-Heap and the $\pm1$RMQ algorithm~\cite{Bender2000}. 
		The tree shows the 2d-Min-Heap of $X=(4, 6, 5, 7, 3, 4, 5, 3)$ represented by arrays $E$ and $D$.
		The gray node $8$ in the tree and gray numbers in the table are added when the last character $X[8]=3$ is processed.
		The boxes with the dashed lines show the answers of RMQs $\rmq{2}{4}$ and $\rmq{5}{7}$.}
	\label{fig:rmq}
\end{figure}

In this section we will describe the algorithm that solves the semi-dynamic RMQ problem
with $O(1)$ query time and amortized $O(1)$ prepend time.
To simplify the algorithm, we consider the prepend operation as appending a character into the end of array. 
In order to solve this problem, Fischer~\cite{ref:conf/wads/Fischer11} proposed an algorithm that uses a 2d-Min-Heap~\cite{ref:journal/SICOMP/Fischer11}
and dynamic LCAs~\cite{ref:journal/SICOMP/Cole05}.
However, the algorithm for dynamic LCAs is very complex to implement.
Therefore, we propose a simple semi-dynamic RMQ algorithm that can be implemented easily if the number of characters to be appended is known beforehand.
This algorithm uses a 2d-Min-Heap and the $\pm1$RMQ algorithm proposed by Bender and Farach-Colton~\cite{Bender2000}.

Let $X$ be a string of length $n$ and let $X[0] = -\infty$.
The 2d-Min-Heap $H$ of $X$ is an ordered tree of $n+1$ nodes $\{0,1,\cdots,n\}$, 
where $0$ is the root node, and the parent node of node $i > 0$ is $\max\{ j < i : X[j] < X[i]\}$.
Moreover, the order of the children is chosen so that they increase from left to right (see Fig.~\ref{fig:rmq} for instance).
Note that the vertices are inevitably aligned in preorder.
Actually, the tree $H$ is represented by arrays $E$ and $D$ that store the sequences of
nodes and their depths visited in an Euler tour of $H$, respectively.
In addition, let $Y$ be an array defined as %that stores the position of the first appearance of $X[i]$ in $E$, i.e., 
$Y[i] = \min\{j : E[j] = i\}$ for each $1 \leq i \leq n$.

For two positions $1 \leq i_1 \leq i_2 \leq n$ in $X$,
$\rmq{i_1}{i_2}$ can be calculated by finding $\lca{i_1}{i_2}$, the LCA of the nodes $i_1$ and $i_2$ in $H$.
If $\lca{i_1}{i_2} = i_1$, then $\rmq{i_1}{i_2} = i_1$.
Otherwise, $\rmq{i_1}{i_2} = i_3$ such that $i_3$ is a child of $\lca{i_1}{i_2}$ and an ancestor of $i_2$.
The $\lca{i_1}{i_2}$ can be computed by performing the $\pm1$RMQ query $\rmqone{Y[i_1]}{Y[i_2]}$ on $D$,
because $D[j+1]-D[j] = \pm1$ for every $j$.
It is known that $\pm1$RMQs can be answered in $O(1)$ time with $O(n)$ time preprocessing~\cite{Bender2000}.
Therefore, we can calculate $\rmq{i_1}{i_2}$ as follows,
\begin{equation*} \label{eq:rmq}
	\rmq{i_1}{i_2} =
	\begin{cases}
		E[\rmqone{Y[i_1]}{Y[i_2]}] & \text{(if $E[\rmqone{Y[i_1]}{Y[i_2]}] = i_1$)}\\
		E[\rmqone{Y[i_1]}{Y[i_2]} + 1] & \text{(otherwise)}.
	\end{cases}
\end{equation*}

Fig.~\ref{fig:rmq} shows an example of calculating the RMQ.
From the property of a 2d-Min-Heap,
arrays $E$ and $D$ are always extended to the end when a new character is appended.
Moreover, the $\pm1$RMQ algorithm can be performed semi dynamically if the size of sequences is known beforehand,
or by increasing the arrays size exponentially. 
Therefore, this algorithm can be performed online and can solve the semi-dynamic RMQ problem, as we intended.
\section{Experimental Results}
In this section, we present experimental results.
We compare the running time of the proposed algorithm in Section~\ref{sec:standard} to
the existing algorithms~\cite{ref:conf/sisap/Benson13,ref:preprint/Pavetic14}.
Furthermore, we show the running time of Algorithm~\ref{alg:OPLCS}.
We used a machine running Ubuntu 14.04 with Core i7 4820K and 64GB RAM.
We implemented all algorithms in {\cpp} and compiled with gcc 4.8.4 with \texttt{-O2} optimization.
We used an implementation of the algorithm proposed by Paveti{\'c}  \emph{et al.}, 
available at \url{github.com/fpavetic/lcskpp}.
We denote the algorithm proposed by Paveti{\'c}~\emph{et al.}~\cite{ref:preprint/Pavetic14} and 
the algorithm proposed by Benson~\emph{et al.}~\cite{ref:conf/sisap/Benson13}  as {\PZS}  and BLMNS, respectively.

\begin{figure}[t]
	\begin{minipage}[t]{.5\linewidth}
		\centering
		\includegraphics[width=\linewidth,clip]{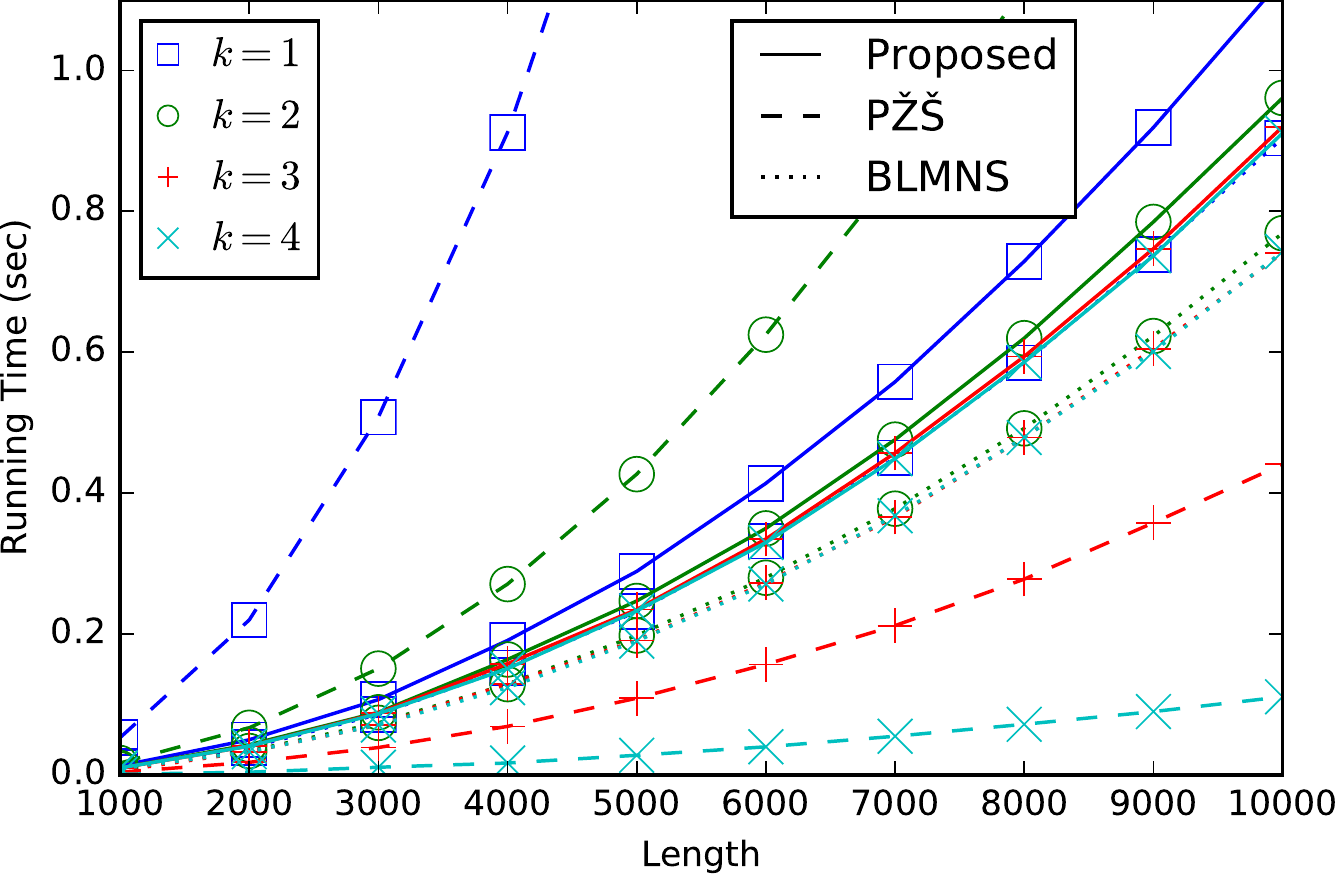}\\
		\vspace{-0.5em} % for setting proper margin between figures
		\subcaption{Random data; $|\Sigma|=4; k=1, 2, 3, 4$}\label{fig:result/LCSk/k}
		\vspace{0.5em} % for setting proper margin between figures
	\end{minipage}%
	\begin{minipage}[t]{.5\linewidth}
		\centering
		\includegraphics[width=\linewidth,clip]{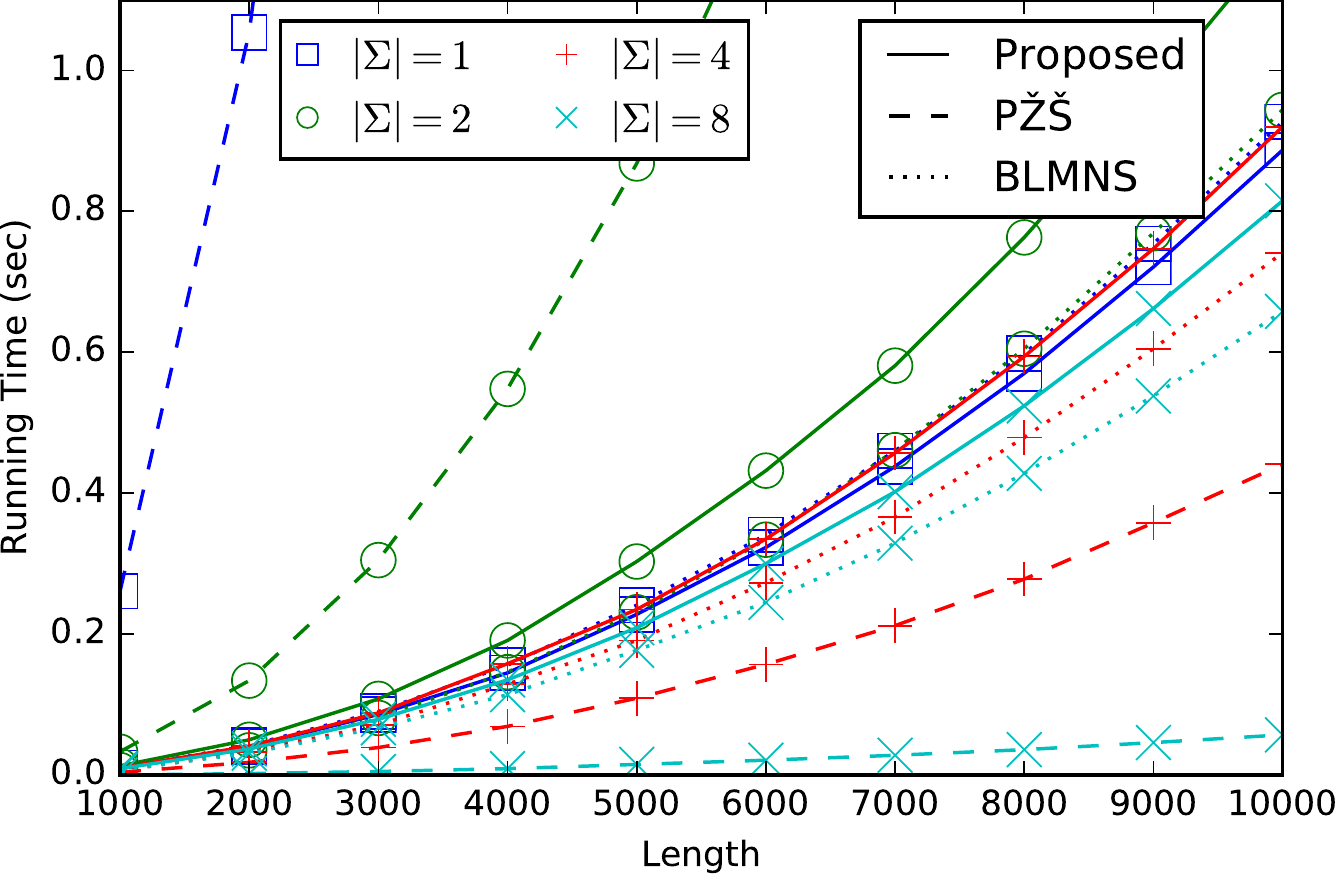}\\
		\vspace{-0.5em}
		\subcaption{Random data; $k=3; |\Sigma| = 1, 2, 4, 8$}\label{fig:result/LCSk/sigma}
		\vspace{0.5em}
	\end{minipage}\\
	\begin{minipage}[t]{.5\linewidth}
		\centering
		\includegraphics[width=.96\linewidth,clip]{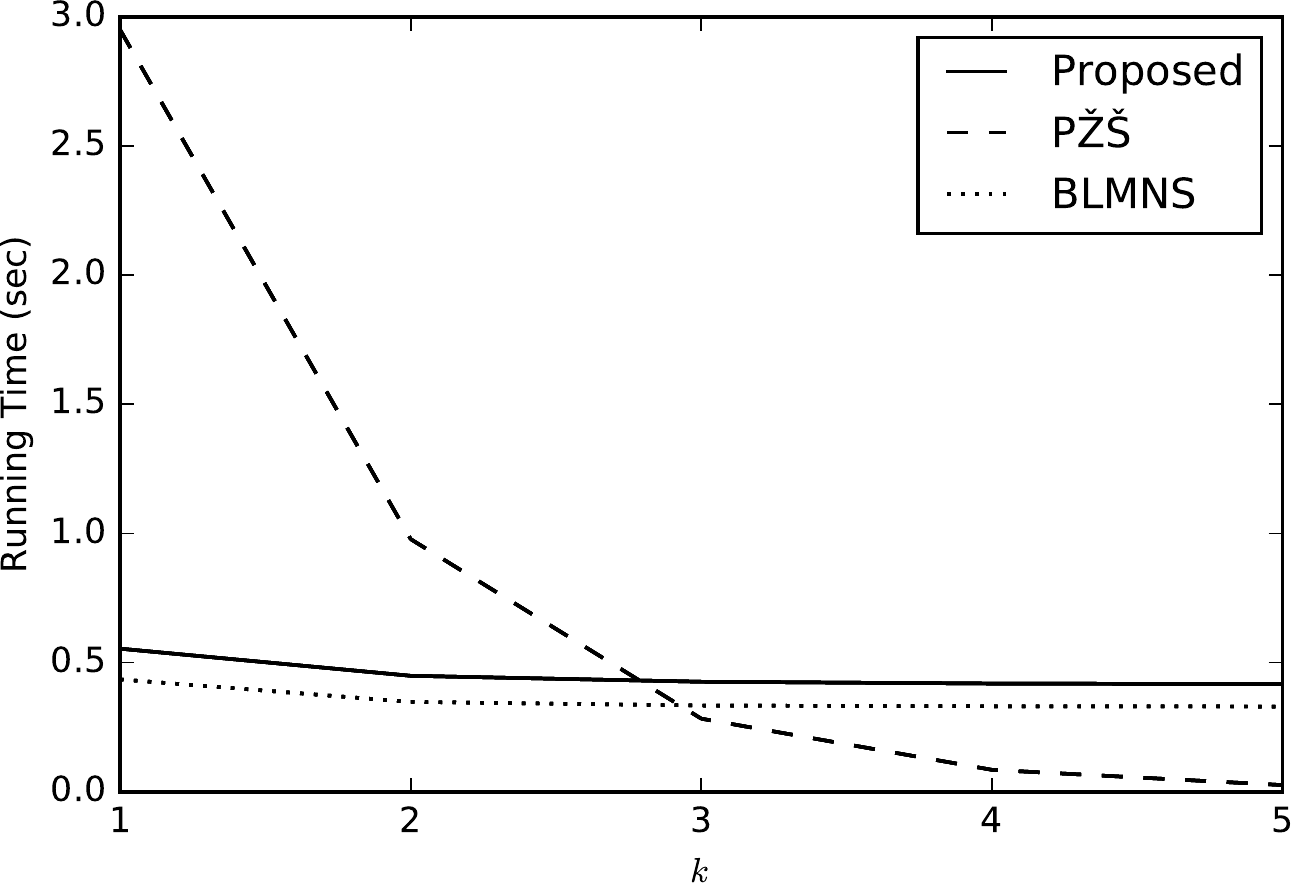}\\
		\vspace{-0.5em}
		\subcaption{DNA data}\label{fig:result/LCSk/DNA}
	\end{minipage}%
	\begin{minipage}[t]{.5\linewidth}
		\centering
		\includegraphics[width=\linewidth,clip]{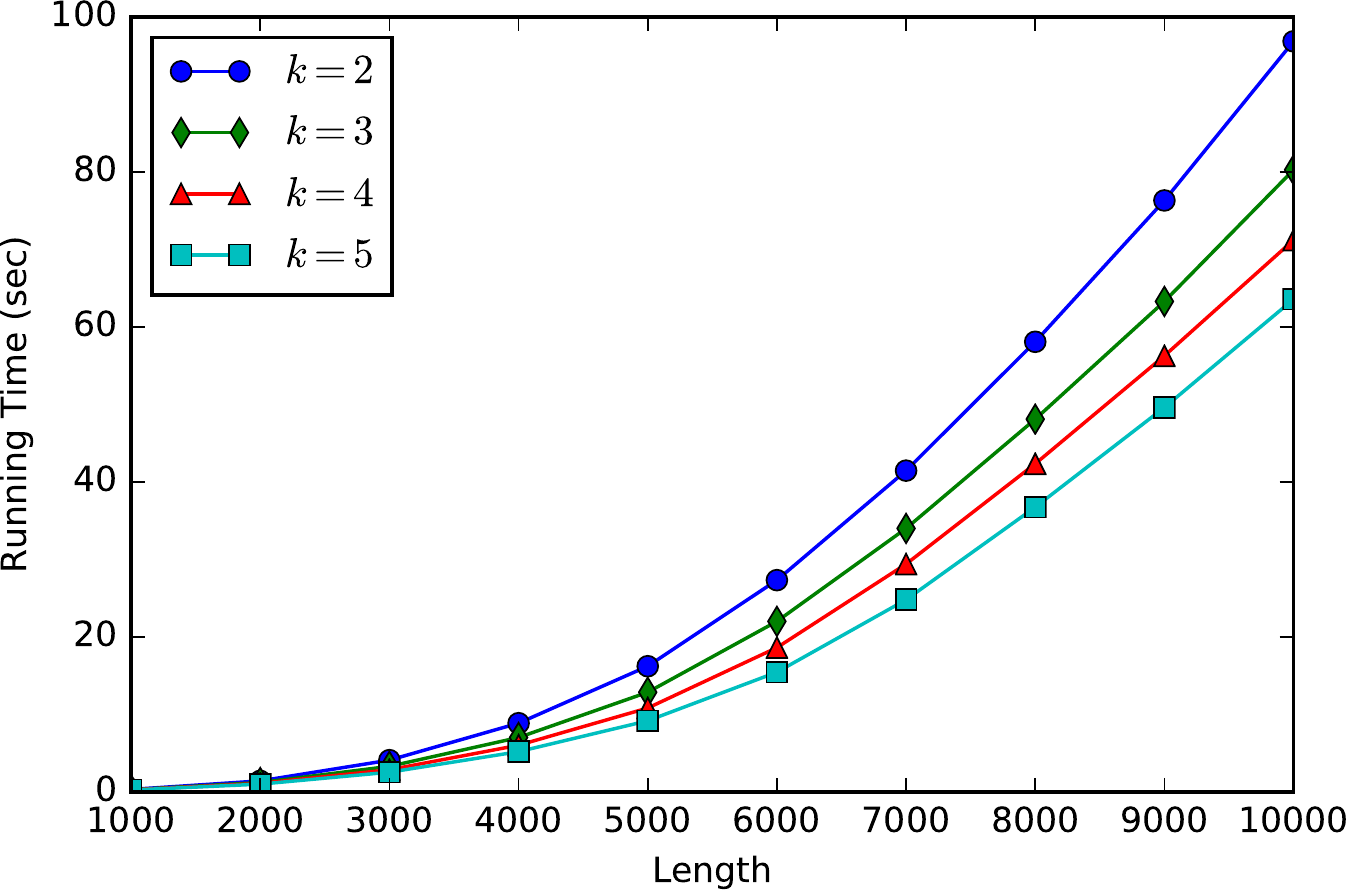}\\
		\vspace{-0.5em}
		\subcaption{Algorithm~\ref{alg:OPLCS}, random data, $|\Sigma| = 100$}\label{fig:result/OPLCS}
	\end{minipage}
	\caption{Running times of the proposed algorithm in Section~\ref{sec:standard}, {\PZS}, and BLMNS~(Figs.\ref{fig:results}\subref{fig:result/LCSk/k}, 
		\ref{fig:results}\subref{fig:result/LCSk/sigma} and \ref{fig:results}\subref{fig:result/LCSk/DNA}),
		and	Algorithm~\ref{alg:OPLCS}~(Fig.~\ref{fig:results}\subref{fig:result/OPLCS}).
		In Figs.~\ref{fig:results}\subref{fig:result/LCSk/k}, \ref{fig:results}\subref{fig:result/LCSk/sigma}, and \ref{fig:results}\subref{fig:result/LCSk/DNA},
		the line styles denote algorithms.
		The line markers in Figs.~\ref{fig:results}\subref{fig:result/LCSk/k} and \ref{fig:results}\subref{fig:result/LCSk/sigma} represent the parameter $k$ and 
		the alphabet size, respectively.
	}\label{fig:results}
\end{figure}

We tested the proposed algorithm in Section~\ref{sec:standard}, {\PZS}, and BLMNS in the following three conditions:
(1) random strings over an alphabet of size $|\Sigma| = 4$ with  $n = m = 1000, 2000, \cdots, 10000$ and $k = 1, 2, 3, 4$
(2) random strings over alphabets of size $|\Sigma|= 1, 2, 4, 8$ with $n = m = 1000, 2000, \cdots, 10000$ and $k = 3$
(3) DNA sequences that are available at \url{www.ncbi.nlm.nih.gov/nuccore/346214858} and \url{www.ncbi.nlm.nih.gov/nuccore/U38845.1},
with $k = 1, 2, 3, 4, 5$.
The experimental results under the conditions (1), (2) and (3) are shown in 
Figs.~\ref{fig:results}\subref{fig:result/LCSk/k}, \ref{fig:results}\subref{fig:result/LCSk/sigma}, and \ref{fig:results}\subref{fig:result/LCSk/DNA}, respectively.

The proposed algorithm in Section~\ref{sec:standard} runs faster than {\PZS} for small $k$ or small alphabets.
This is due to that {\PZS} strongly depends on
the total number of matching $k$ length substring pairs between input strings,
and for small $k$ or small alphabets there are many matching pairs.
In general BLMNS runs faster than ours.
The proposed algorithm runs a little faster for small $k$ or small alphabets, except $|\Sigma| = 1$.
We think that this is because
for small $k$ or small alphabets the probability that $L[i, j] \ge k$ is high,
and this implies that we need more operations to compute $M[i, j]$ by definition.
In Fig.~\ref{fig:results}\subref{fig:result/LCSk/sigma}, it is observed that the proposed algorithm with $|\Sigma| = 1$ runs faster
than with $|\Sigma| = 2$.
Since $|\Sigma| = 1$ implies that $X = Y$ if $X$ and $Y$ have the same length,
$L[i, j] > k$ almost always holds,
which leads to reduce branch mispredictions and speed up execution.

We show the running time of Algorithm~\ref{alg:OPLCS} in Fig.~\ref{fig:results}\subref{fig:result/OPLCS}.
We tested Algorithm~\ref{alg:OPLCS} on random strings over $\Sigma = \{1, 2, \cdots, 100\}$ with $n=m=1000, 2000, \cdots, 10000$ and $k = 2, 3, 4, 5$.
It is observed that the algorithm runs faster as the parameter $k$ is smaller.
We suppose that the hidden constant of the RMQ data structure described in Section~\ref{sec:rmq} is large.
Therefore, the running time of Algorithm~\ref{alg:OPLCS} depends on 
the number of times the \texttt{rmq} operation is called,
and for small $k$ the number of them increases since the probability that $l \ge k$ is high.

\section{Conclusion}\label{sec:conclusion}
We showed that both the {\problemabbr} problem and the {\opproblemabbr} problem can be solved in $O(mn)$ time.
Our result on the {\problemabbr} problem gives a better worst-case running time than previous algorithms~\cite{ref:conf/sisap/Benson13,ref:preprint/Pavetic14},
while the experimental results showed that the previous algorithms run faster than ours on average.
Although the {\opproblemabbr} problem looks much more challenging than the {\problemabbr},
since the former cannot be solved by a simple dynamic programming due to the properties of order-isomorphisms,
the proposed algorithm achieves the same time complexity as the one for the {\problemabbr}.

\subsubsection*{Acknowledgements.}
This work was funded by ImPACT Program of Council for Science, Technology and 
Innovation (Cabinet Office, Government of Japan), 
Tohoku University Division for Interdisciplinary Advance Research and Education,
and JSPS KAKENHI Grant Numbers JP24106010, JP16H02783, JP26280003.

\bibliographystyle{abbrv}
\bibliography{ref}

\clearpage
\appendix
\section*{Appendix}
Springer's version (\url{http://dx.doi.org/10.1007/978-3-319-51963-0_28}) contains some crucial typos
which were inserted during Springer's typesetting process.

\begin{itemize}
	\item p.~364, Eq (4): \mbox{}\\
		\begin{equation*}
			\opLCE{i_1}{i_2} = \min\bigl\{\Zi{{\left(S_1 \cdot S_2\right)}{i_1}}{|S_1| - i_1 + i_2 + 1}, \ |S_1| - i_1 + 1 \bigr\}
		\end{equation*}
	should be replaced with
		\begin{equation*}
			\opLCE{i_1}{i_2} = \min\bigl\{\Zi{\suffix{\left(S_1 \cdot S_2\right)}{i_1}}{|S_1| - i_1 + i_2 + 1}, \ |S_1| - i_1 + 1 \bigr\}
		\end{equation*}
	
	\item p.~368, paragraph 5, line 8: \mbox{}\\
	\begin{quote}
		\vspace{-3ex}
		Therefore, if we can compute tables $\Prev{{S}{i}}$ and $\Next{{S}{i}}$
	\end{quote}
	should be replaced with
	\begin{quote}
		Therefore, if we can compute tables $\Prev{\suffix{S}{i}}$ and $\Next{\suffix{S}{i}}$		
	\end{quote}

	\item p.~368, paragraph 5, line 9: \mbox{}\\
	\begin{quote}
		\vspace{-3ex}
		with $O(|S|\log|S|)$ time preprocessing, $\Z{{S}{i}}$ for		
	\end{quote}
	should be replaced with
	\begin{quote}
		with $O(|S|\log|S|)$ time preprocessing, $\Z{\suffix{S}{i}}$ for		
	\end{quote}
	
	\item p.~369, paragraph 1, line 1: \mbox{}\\
	\begin{quote}
		\vspace{-3ex}
		In order to compute the tables $\Prev{{S}{i}}$ and $\Next{{S}{i}}$
	\end{quote}
	should be replaced with
	\begin{quote}
		In order to compute the tables $\Prev{\suffix{S}{i}}$ and $\Next{\suffix{S}{i}}$
	\end{quote}

	\item p.~369, paragraph 1, line 5: \mbox{}\\
	\begin{quote}
		\vspace{-3ex}
		We can compute $\Prev{{S}{i}}$ and $\Next{{S}{i}}$ for each $1 \le i \le \len{S}$
	\end{quote}
	should be replaced with
	\begin{quote}
		We can compute $\Prev{\suffix{S}{i}}$ and $\Next{\suffix{S}{i}}$ for each $1 \le i \le \len{S}$
	\end{quote}
	
\end{itemize}

\end{document}